\documentclass[11pt]{article}
\usepackage{latexsym}
\usepackage{amsmath}
\usepackage{amssymb}
\usepackage{amsthm}
\usepackage{epsfig}
\usepackage{algorithm,algorithmic}
\usepackage[mathscr]{euscript}
\usepackage[sort]{cite}
\usepackage{wrapfig}
\usepackage{color}
\usepackage{authblk}
\usepackage{hyperref}
\hypersetup{
    colorlinks=true, 
    linktoc=all,     
    linkcolor=blue,  
    citecolor=blue,
    urlcolor=blue
}

\usepackage{multirow}
\usepackage{footnote}
\makesavenoteenv{tabular}
\makesavenoteenv{table}

\newtheorem{theorem}{Theorem}

\newtheorem{lemma}[theorem]{Lemma}

\newcommand{\expec}[1]{\mathbb E\left [ #1 \right ]}
\newcommand{\var}[1]{\mathbb V\left [ #1 \right ]}
\newcommand{\covar}[1]{\mathbb{COV}\left [ #1 \right ]}

\newcommand{\prob}[1]{\mathbb P \left [ #1 \right ]}

\DeclareMathOperator{\polylog}{polylog}
\DeclareMathOperator{\poly}{poly}

\DeclareMathOperator{\oracle}{{oracle}}

\newcommand{\eat}[1]{}

\newcommand{\bin}{{\mathbf{Bin}}}

\renewcommand{\H}{{\mathcal H}}

\newcommand{\wt}{\widetilde}
\newcommand{\wh}{\widehat}
\newcommand{\light}{\texttt{L}}
\newcommand{\heavy}{\texttt{H}}

\advance \topmargin by -\headheight
\advance \topmargin by -\headsep
\evensidemargin \oddsidemargin
\begin{document}
\title{Improved 3-pass Algorithm for Counting 4-cycles\\in Arbitrary Order Streaming\footnote{This work is supported by NSF Award \#1907738.}}
\author{Sofya Vorotnikova}
\affil{Dartmouth College\\
\texttt{svorotni@gmail.com}}
\date{}
\maketitle

\abstract{The problem of counting small subgraphs, and specifically cycles, in the streaming model received a lot of attention over the past few years. In this paper, we consider arbitrary order insertion-only streams, improving over the state-of-the-art result on counting 4-cycles. Our algorithm computes a $(1+\epsilon)$-approximation by taking three passes over the stream and using space $O(\frac{m \log n}{\epsilon^2 T^{1/3}})$, where $m$ is the number of edges in the graph and $T$ is the number of 4-cycles.}

\section{Introduction}
Subgraph counting is a fundamental graph problem and an important primitive in massive graph analysis. It has many applications in data mining and analyzing the structure of large networks. This problem has also received a lot of attention in the streaming community, with the main focus on counting triangles~\cite{AhnGM2012a, Bar-YKS2002, BOV13, BurioFLMS2006, JowhaG2005, MVV16, CormodeJ14, KutzkovP14a, PavanTTW13, BC17, KallaugherMPV19, MV20}.
Several papers considered counting larger cycles and cliques~\cite{BC17,KallaugherMPV19, ManjunathMPS11},
and a few studied arbitrary subgraphs of constant size~\cite{BC17,KP17,Kane12}. 
There is also work on counting 4-cycles in the case when the underlying graph is bipartite~\cite{butt}. Since a 4-cycle is also a 2-by-2 biclique, it is the most basic motif in bipartite graphs and plays essentially the same role as a triangle does in general graphs.

In this paper, we concentrate on counting 4-cycles in the arbitrary order insertion-only streaming model, improving over the state-of-the-art algorithm presented by McGregor and Vorotnikova~\cite{MV20}.

\subsection{Our Result and Previous Work}
Throughout this paper, we use $n$ to denote the number of vertices in the graph, $m$ to denote the number of edges, and $T$ for the number of 4-cycles. Note that our algorithm is parameterized in terms of $T$, which is a convention adopted in the literature. In practice, the quantities in the algorithm would be initialized based on a promised lower bound on $T$.  

Our result is as follows.
\begin{theorem}
	There exists an $O(\frac{m \log n}{\epsilon^2 T^{1/3}})$ space algorithm that takes three passes over an arbitrary order stream and returns a $(1+\epsilon)$ multiplicative approximation to the number of 4-cycles in the graph with probability at least $1/4$. 
\end{theorem}
By running $\Theta(\log 1/\delta)$ copies of the algorithm in parallel and taking the median of their outputs, we can increase the success probability to $1 - \delta$, where $\delta \in (0, 1)$.

Our algorithm can be directly compared to the $\wt{O}(m/T^{1/4})$ space\footnote{We use $\wt{O}(\cdot)$ notation to hide $\polylog(n)$ and $1/\epsilon$ factors.} algorithm by McGregor and Vorotnikova~\cite{MV20}. It takes the same number of passes over the stream and has the same approximation guarantees. We believe that the space of our algorithm is tight, however the best known lower bound is currently $\Omega(m/T^{1/2})$~\cite{MV20}.

In~\cite{BC17} Bera and Chakrabarti present a different 4-cycles counting algorithm which takes four passes and uses space $\wt{O}(m^2/T)$. Note, that the space used by our algorithm is as good or better when $T = O(m^{3/2})$. McGregor and Vorotnikova~\cite{MV20} also present a 2-pass $\wt{O}(m^{3/2}/T^{3/4})$ space algorithm which distinguishes between graphs with 0 and $T$ 4-cycles.

\section{Algorithm and Analysis}
\subsection{Notation}
A \textit{wedge} is a path of length 2. For wedge $(u, b, v)$ we call vertices $u$ and $v$ the \textit{endpoints} of the wedge and vertex $b$ the \textit{center}.

We use $\Gamma(v)$ to denote the set of neighbors of vertex $v$. Consider sets of vertices $\{u,v\}$ and $\Gamma(u) \cap \Gamma(v)$. Edges between these two sets form a complete bipartite graph, which we call a \textit{diamond} with endpoints $u$ and $v$. We say that wedge $w$ is a part of diamond $d$ if they have the same endpoints. Note that a diamond with endpoints $u$ and $v$ consists of $|\Gamma(u) \cap \Gamma(v)|$ wedges and involves ${|\Gamma(u) \cap \Gamma(v)| \choose 2}$ 4-cycles.

Throughout the paper we use $t(e)$, $t(w)$, and $t(d)$ to denote the number of 4-cycles involving edge $e$, wedge $w$, or involved in diamond $d$ respectively. For any quantity $k$, we use $\wh{k}$ to denote its estimate.

In Section~\ref{alg}, we define heavy/light edges, wedges, and diamonds, where ``heavy'' roughly corresponds to ``involved in many 4-cycles'' and ``light'' to ``involved in few 4-cycles''. Note that these are defined by the algorithm and depend on the collected samples of vertices and edges. We define $T_H$ to be the number of 4-cycles with at least one heavy wedge and $T_L$ as the number of 4-cycles with no heavy wedges and at most one heavy edge.

\subsection{Main Idea}
The most basic algorithm approximating the number of 4-cycles in a graph is as follows:
\begin{description}\parskip=0ex
\item[Pass 1:] Sample edges with probability $p$, call set $S$.
\item[Pass 2:] For each edge $e$ in the stream, let $s(e)$ be the number of 3-paths with all edges in $S$ that $e$ completes to a 4-cycle.
\item[Return:]  $\frac{1}{4p^3}\sum_{e \in E} s(e)$.
\end{description}
In expectation, the value returned by this algorithm is $T$. However, due to the fact that some edges or wedges in the graph can be involved in a large number of 4-cycles, the variance of this estimator is large. If an edge or wedge participates in many 4-cycles, call it ``bad''. In this paper, we show that it is possible to identify such bad edges and wedges and take care of them separately, leading to an accurate approximation.

We observe that if wedge $(u, b, v)$ is bad, then it is a part of a large diamond with endpoints $u$ and $v$. If we sample $\wt{\Omega}(1)$ vertices in $\Gamma(u) \cap \Gamma(v)$ and collect all incident edges, we will detect the diamond and accurately estimate its size. Using this method, we approximate the total number of cycles with bad wedges.

We then separately approximate the number of cycles with no bad wedges and at most one bad edge. This procedure follows the same template as the arbitrary order 4-cycle counting algorithm in~\cite{MV20}. Sampling edges uniformly at a certain rate allows us to obtain some 3-paths which are involved in 4-cycles with no bad wedges. Additionally, sampling vertices uniformly and storing all incident edges allows us to build an oracle roughly classifying edges as good or bad. We use this oracle to compute the number of bad edges in each of the cycles we discover. Note that the oracle takes an extra pass over the stream, and thus in total our algorithm uses three passes.

\subsection{Algorithm} \label{alg}
The algorithm in this section computes estimates to $T_H$ and $T_L$ separately and then returns their sum. We later show that $\wh{T}_H + \wh{T}_L$ is an accurate approximation of $T$. 

Within the algorithm, we define heavy/light diamonds and wedges. Roughly speaking, a heavy diamond consist of $\Omega(T^{1/3})$ wedges and a light diamond consist of $O(T^{1/3})$ wedges. A wedge is then defined as heavy or light if it is a part of a heavy or light diamond respectively.

In the third pass, we refer to the oracle which classifies edges as heavy or light. It is described separately after the main algorithm.
\begin{description}\parskip=0ex
\item[Pass 1:] \strut
\begin{itemize}
\item Let $p = \frac{c \log n}{ \epsilon^2 T^{1/3}}$.
\item Sample edges with probability $p$, call set $S_E$.
\item Sample vertices with probability $p$, call set $Q_V$. Collect all incident edges, call set~$Q_E$.
\item Sample vertices with probability $p$, call set $Z_V$. Collect all incident edges, call set $Z_E$.
\end{itemize}
\item[After Pass 1:] \strut
\begin{itemize}
\item For a pair of vertices $(u, v)$, let $q(u,v)$ be the number of wedges with center in $Q_V$ and endpoints $u$ and $v$. \item Define diamond $d$ with endpoints $u$ and $v$ to be \textit{heavy} if $q(u,v) \geq p T^{1/3}$ and \textit{light} otherwise. Let $\wh{t}(d) = {q(u,v)/p \choose 2}$.
\item Define wedge $w$ with endpoints $u$ and $v$ to be \textit{heavy} if it is part of a heavy diamond and \textit{light} otherwise. Let $\wh{t}(w) = q(u,v)/p - 1$.
\item Find all pairs of vertices $(u,v)$ which are endpoints of heavy diamonds/wedges.
\item Let $\wh{T}_H = \sum \wh{t}(d)$, where $d$ is a heavy diamond.
\end{itemize}
\item[Pass 2:] For every edge $e$ in the stream:
\begin{itemize}
\item Check if $e$ completes any 3 edges from $S_E$ to a 4-cycle (call it $\tau$). Check whether $\tau$ has a heavy wedge; if not, store $(e, \tau)$.
\end{itemize}
\item[Pass 3:] \strut
\begin{itemize}
\item For all edges involved in cycles stored in pass 2, use $\oracle(Z_V, Z_E)$ to classify them as heavy or light.
\item Let $A_0$ be the number of $(e, \tau)$ pairs s.t. $\tau$ has no heavy edges.
\item Let $A_1$ be the number of $(e, \tau)$ pairs s.t. $e$ is heavy and the other 3 edges in $\tau$ are light.
\item Let $\wh{T}_L = A_0/(4p^3) + A_1/p^3$
\end{itemize}
\item[Return:] $\wh{T}_H + \wh{T}_L$
\end{description}

\paragraph{Oracle.}
Below, we describe the oracle which classifies edges as heavy or light. Roughly speaking, heavy edges are involved in $\Omega(T^{2/3})$ 4-cycles and light edges in $O(T^{2/3})$.

Suppose, that we need to classify edge $e = (u,v)$ as heavy or light.  We then look at edges sharing a vertex with $e$. In the post-processing of the first pass, we determined all pairs of vertices which are endpoints of heavy diamonds/wedges. Thus, for wedge $(e, e')$ we can refer to that list to check whether it is heavy or not. If it is heavy, we also get an estimate of the number of 4-cycles it is involved in and thus contributes to $t(e)$. Separately, we approximate the total number of 4-cycles on $e$ which involve two light wedges $(e, e')$ and $(e, e'')$.

\begin{description}\parskip=0ex
\item[oracle($Z_V$, $Z_E$, $e$):] \strut
\begin{itemize}
\item Let $\wh{t}_H(e) \leftarrow 0$ and $\wh{t}_L(e) \leftarrow 0$.
\item For wedges of the form $(e, e')$, where $e' \in Z$: if $(e, e')$ is heavy, ``exclude''\footnote{When we talk about ``excluding'' edges from $Z$, we need to ``exclude'' different sets of edges for different instances of the oracle. In practice, for each instance mark those edges and ignore them. However, they might be used by other instances.} $e'$ from $Z$.
\item For each edge $e^*$ in the stream, s.t. $e^*$ shares a vertex with $e$: 
\begin{itemize}
\item Look up whether $(e, e^*)$ is heavy.
\item If heavy,  $\wh{t}_H(e) \leftarrow \wh{t}_H(e) + \wh{t}(e, e^*)$.
\item If light and $e^* = (v, a)$, let $\lambda(e, e^*)$ be the number of vertices $b \in Z_V$, such that $(u, v, a, b)$ is a 4-cycle. $\wh{t}_L(e) \leftarrow \wh{t}_L(e) + \lambda(e, e^*)/p$.
\end{itemize}
\item Let $\wh{t}(e) = \wh{t}_H(e) + \wh{t}_L(e)$.
\item Return:
$
\begin{cases}
\light & \mbox{ if } \wh{t}(e) < T^{2/3} \mbox{ (light edge)}\\
\heavy & \mbox{ if } \wh{t}(e) \geq T^{2/3} \mbox{ (heavy edge)}
\end{cases}
$
\end{itemize}
\end{description}

\subsection{Correctness}
\subsubsection{Oracle}
In Lemma~\ref{oracle_lemma}, we show that light edges are involved in at most $4 T^{2/3}$ 4-cycles and heavy edges are involved in at least $T^{2/3}/4$ cycles. Note that the oracle relies on the procedure estimating the number of 4-cycles on a heavy wedge, so in the proof we refer to Lemma~\ref{heavy_wedges} below.
\begin{lemma}\label{oracle_lemma}
With high probability
\begin{description}
\item[a.] $\oracle(Z_V, Z_E, e) = \light$ implies $t(e) \leq 4 T^{2/3}$
\item[b.] $\oracle(Z_V, Z_E, e) = \heavy$ implies $t(e) \geq T^{2/3}/4$
\end{description}
\end{lemma}

\begin{proof}
Let $t_H(e)$ be the number of 4-cycles on $e$, where $e$ is a part of a heavy wedge. Let $t_L(e) = t(e) - t_H(e)$. Let $\wh{t}_H(e)$ and $\wh{t}_L(e)$ be our estimates of those two quantities. 

Note that in the process of approximating $t_H(e)$, we are double-counting 4-cycles with two heavy wedges involving $e$. However, we can show that this double-count is negligible. Let $D(e)$ be the number of heavy diamonds which involve $e$. Since each 4-cycle can belong to at most 2 diamonds, we are double-counting at most $D(e)^2$ cycles. From Lemma~\ref{heavy_wedges} part \textbf{(b)}, it follows that the number of 4-cycles in a heavy diamond is at least ${T^{1/3}/2 \choose 2} \geq T^{2/3}/9$. Therefore, $t_H(e) \geq D(e) T^{2/3}/9$ and $D(e) \leq 2T/(T^{2/3}/9) \leq 5T^{1/3}$. If $T$ is sufficiently large, then $D(e)^2 < (\epsilon/2)t_H(e)$. 

From Lemma~\ref{heavy_wedges} part \textbf{(c)} it follows that
$$\sum_{\substack{\textrm{heavy } w:\\e \in w}} \wh{t}(w) = (1 \pm \epsilon/2) \sum_{\substack{\textrm{heavy } w:\\e \in w}} t(w)$$
Taking double-counting into account,
\begin{equation} \label{oracle_eq1}
\wh{t}_H(e) = (1 \pm \epsilon) t_H(e)
\end{equation}

Recall that $e = (u, v)$ and let $X_b$ be the number of cycles $(u,v,a,b)$ with no heavy wedges if $b \in Z_V$, and $0$ otherwise. Let $X_L = \sum_{b \in V} X_b$ and note that $\expec{X_L} = p t_L(e) = p \expec{\wh{t}_L(e)}$.

If $t_L(e) < T^{2/3}$, then $\expec{\wh{t}_L(e)} < T^{2/3}$, and from the Chernoff bound it follows that
\begin{equation} \label{oracle_eq2}
\prob{|\wh{t}_L(e) - t_L(e)| \geq T^{2/3}/4} 
= \prob{|X_L - p t_L(e)| \geq p T^{2/3}/4}
\leq 2 \exp \left(- \frac{16 p T^{2/3}}{6 T^{1/3}} \right)
\leq 1/\poly(n)
\end{equation}
where the first inequality follows from the fact that $X_b \leq 2T^{1/3}$ for all $b$. Similarly, if $t(e) \geq T^{2/3}$, then 
\begin{equation} \label{oracle_eq3}
\prob{|\wh{t}_L(e) - t_L(e)| \geq t(e)/4} \leq 1/\poly(n)
\end{equation}
We first prove the contrapositive of \textbf{(a)}. Assume $t(e) > 4 T^{2/3}$. Then from Eq.~\ref{oracle_eq1} (taking  $\epsilon = 1/4$) and Eq.~\ref{oracle_eq3},
$$
\wh{t}(e) \geq (t_H(e) - t_H(e)/4) + (t_L(e) - t(e)/4)
\geq t(e) - t(e)/2 
> T^{2/3} 
$$
Similarly, we prove the contrapositive of \textbf{(b)} from Eq.~\ref{oracle_eq1} and~\ref{oracle_eq2}. If $t(e) < T^{2/3}/4$, then 
$$
\wh{t}(e) \leq (t_H(e) + t_H(e)/4) + (t_L(e) + T^{2/3}/4) < T^{2/3}
$$
\end{proof}

\subsubsection{Estimating $T_H$}
In Lemma~\ref{heavy_wedges}, we prove that we can distinguish between large and small diamonds and estimate the number of 4-cycles in a heavy diamond or on a heavy wedge.
\begin{lemma}\label{heavy_wedges}
Let $w(d)$ be the number of wedges in diamond $d$, and let $q(d)$ be the number of those wedges with center in $Q_V$. Recall that $t(d)$ is the number of 4-cycles in diamond $d$, and $t(w)$ is the number of 4-cycles involving wedge $w$. Then with high probability,
\begin{description}
\item[a.] If diamond $d$ is heavy ($q(d) < pT^{1/3}$), then $w(d) \leq 2T^{1/3}$
\item[b.] If diamond $d$ is light ($q(d) \geq pT^{1/3}$), then $w(d) \geq T^{1/3}/2$
\item[c.] If wedge $w$ is heavy, then $\wh{t}(w) = q(d)/p - 1 = (1 \pm \epsilon/2) t(w)$
\item[d.] If diamond $d$ is heavy, then $\wh{t}(d) = {q(d)/p \choose 2} = (1 \pm \epsilon/4) t(d)$
\end{description}
\end{lemma}

\begin{proof}
Observe that $q(d) \sim \bin(w(d), p)$. By an application of the Chernoff bound, if $w(d) \geq 2T^{1/3}$, then 
$$\prob{q(d) < pT^{1/3}} \leq \exp(-2pT^{1/3}/3) \leq 1/\poly(n)$$
proving \textbf{(a)}. Statement \textbf{(b)} is proved similarly.

Note that the number of 4-cycles in a diamond grows as the square of the number of wedges. Therefore, to get a $(1 + \epsilon/4)$-approximation to $t(d)$, we need to estimate $w(d)$ to a higher accuracy. If $q(d) \geq pT^{1/3}$, from Chernoff it follows that
$$\prob{|q(d) - w(d) p| \geq (\epsilon/20) w(d) p} \leq 2 \exp(- \epsilon^2 w(d) p /1200) \leq 1/\poly(n)$$
Recall that if a diamond consists of $k$ wedges, then the number of 4-cycles on each of those wedges is $k-1$. Therefore, statement \textbf{(c)} follows since $\epsilon/20 < \epsilon/2$. Statement \textbf{(d)} follows since ${(1+\epsilon/20)w(d) \choose 2} \leq (1 + \epsilon/4){w(d) \choose 2}$ and ${(1-\epsilon/20)w(d) \choose 2} \geq (1 - \epsilon/4){w(d) \choose 2}$.
\end{proof}

\begin{lemma} \label{heavy_lemma}
With high probability, $\wh{T}_H = T_H \pm \epsilon T/3$.
\end{lemma}
\begin{proof}
First, note that our algorithm double-counts 4-cycles which are involved in two heavy diamonds. As was mentioned before, the number of 4-cycles in a heavy diamond is at least $T^{2/3}/9$, and thus the number of heavy diamonds is at most $5T^{1/3}$. Since two diamonds can have at most one cycle in common, we are double-counting at most $25 T^{2/3} \leq (\epsilon/12)T$ cycles. The rest of the proof follows from Lemma~\ref{heavy_wedges} part \textbf{(d)}.
\end{proof}

\subsubsection{Estimating $T_L$}
\begin{lemma} \label{light_lemma}
With constant probability, $\wh{T}_L = T_L \pm \epsilon T/2$.
\end{lemma}
\begin{proof}
Let $T_i$ be the number of 4-cycles in $T_L$ with $i$ heavy edges. Let $\wh{T}_0 = A_0/(4p^3)$ and $\wh{T}_1 = A_1/p^3$. Note that $\expec{\wh{T}_0} = T_0$ and $\expec{\wh{T}_1} = T_1$.

We now show that with constant probability, $\wh{T}_0 = T_0 \pm \epsilon T/4$ and $\wh{T}_1 = T_1 \pm \epsilon T/4$.

By an application of the Chebyshev bound, 
$$\prob{|\wh{T}_0 - T_0| \leq \epsilon T/4} \leq 1/16$$
as long as $\var{\wh{T}_0} \leq \epsilon^2 T^2/256$. We now give a bound on the variance of $\wh{T}_0$. Let $\H_0$ be the set of 3-paths which are involved in 4-cycles in $T_0$.  Let $X_q$ be 1 if all 3 edges of path $q \in \H_0$ were sampled and 0 otherwise. Then 
\begin{align}
\var{\wh{T_0}} &= \var{\frac{1}{4p^3} \sum_{q \in \H_0} X_q} \nonumber\\
&= \frac{1}{16 p^6} \left( \sum_{q \in \H_0}\var{X_q} + \sum_{\substack{q, t \in \H_0 \colon \\ q \neq t, \\ q \cap t \neq \emptyset}} \covar{X_q, X_t} \right) \nonumber\\
& \leq \frac{1}{16 p^6} \left( \sum_{q \in \H_0}\expec{X^2_q} + \sum_{\substack{q, t \in \H_0 \colon \\ q \neq t, \\ q \cap t \neq \emptyset}} \expec{X_q X_t} \right) \nonumber\\
& \leq \frac{1}{16 p^6} \left( \sum_{q \in \H_0}p^3 + \sum_{q \in \H_0} \sum_{\substack{t \in \H_0 \colon \\ q \neq t, \\ |q \cap t| = 1}} p^5 + \sum_{q \in \H_0} \sum_{\substack{t \in \H_0 \colon \\ q \neq t, \\ |q \cap t| = 2}} p^4 \right) \nonumber\\
& \leq \frac{1}{16 p^6} \left( |\H_0| p^3 + \sum_{q \in \H_0} cT^{2/3} p^5 + \sum_{q \in \H_0} cT^{1/3} p^4 \right) \label{eq1}\\
& \leq \frac{1}{16 p^6} \left(|\H_0| p^3 + c|\H_0|T^{2/3} p^5 + c|\H_0|T^{1/3} p^4 \right) \nonumber\\
&\leq T/4p^3 + cT^{5/3}/4p + cT^{4/3}/4p^2 \nonumber\\
&\leq \epsilon^2 T^2/256 \label{eq2}
\end{align}

Equation~\ref{eq1} follows from the fact that any path $q \in \H_0$ intersects at most $12 T^{2/3}$ other paths in $\H_0$ at one edge and at most $4 T^{1/3}$ paths at two edges (from Lemmas~\ref{oracle_lemma} and \ref{heavy_wedges}). Equation~\ref{eq2} follows from our definition of $p$.

Proving $\prob{|\wh{T}_1 - T_1| \leq \epsilon T/4} \leq 1/16$ follows along the same lines.
\end{proof}

\subsubsection{Estimating $T$}
We refer to one of the lemmas in~\cite{MV20}, which bounds the number of 4-cycles with at most one edge which is involved in a lot of cycles.
\begin{lemma}[McGregor and Vorotnikova~\cite{MV20}]\label{old_lemma}
We call an edge $e$ ``bad'' if it is contained in at least $\eta \sqrt{T}$ 4-cycles, and ``good'' otherwise.  There are at least $(1-82/\eta)T$ cycles containing no more than one bad edge.
\end{lemma}

Applying this lemma with $\eta  = T^{1/6}/4$, we get that the number of cycles with at most one bad edge is at least $(1-328/T^{1/6})T \leq (1- \epsilon/6)T$. 
We can now prove the main lemma.

\begin{lemma}
With constant probability, $\wt{T} =  (1 \pm \epsilon)T$. 
\end{lemma}
\begin{proof}
Let $T'_H$ be the number of cycles with at least one heavy wedge and at most one heavy edge. Note that $T'_H \leq T_H$. Since good edges (with $t(e) \leq T^{2/3}/4$) are classified as light w.h.p., 
$$(1- \epsilon/6)T \leq T_L + T'_H \leq T_L + T_H \leq T$$
where the first inequality follows from Lemma~\ref{old_lemma}. The rest of the proof follows from Lemmas~\ref{heavy_lemma} and~\ref{light_lemma}.
\end{proof}

\subsection{Space analysis}
Sets $S_E$, $Q_E$, and $Z_E$ all have the same expected size $mp = O(\frac{m \log n}{\epsilon^2 T^{1/3}})$. The expected number of cycles stored in pass 2 is $4T/p^3 = \wt{O}(1)$. Finally, the extra space used by each instance of $\oracle(Z_V, Z_E, e)$ is in expectation $\wt{O}(1)$, since it keeps track of a constant number of counters and $\wt{O}(1)$ ``excluded'' edges, corresponding to heavy wedges involving $e$ among the input of the instance. Therefore, the total space used by the algorithm is $O(\frac{m \log n}{\epsilon^2 T^{1/3}})$.

\bibliographystyle{plain}
\bibliography{references}

\end{document}